\documentclass[aps,prd,reprint,superscriptaddress,nofootinbib,amsmath,amssymb,floatfix]{revtex4-2}
\usepackage[T1]{fontenc}
\usepackage{graphicx}
\usepackage{bm}
\usepackage{booktabs}
\usepackage{tikz}
\usepackage{verbatim}
\usetikzlibrary{arrows.meta,calc}
\usepackage[colorlinks=true,linkcolor=blue,citecolor=blue,urlcolor=blue]{hyperref}
\usepackage{amsmath,amssymb,amsthm}
\newtheorem{theorem}{Theorem}
\newtheorem{lemma}{Lemma}
\newtheorem{corollary}{Corollary}

\newcommand{\ind}{\operatorname{ind}}

\begin{document}

\title{A No-Go Theorem for Nonspherical Bounces with a Single Negative Mode}

\affiliation{School of Data Science and Artificial Intelligence, Wenzhou University of Technology, Wenzhou, China}

\author{Guangshang Chen}
\affiliation{
Key Laboratory of Particle Physics and Particle Irradiation (MOE), Institute of Frontier and Interdisciplinary Science,
Shandong University, Qingdao 266237, P. R. China
}

\author{Jianpeng Dai}
\email{Corresponding author, daijp1811833596@gmail.com}
\affiliation{School of Data Science and Artificial Intelligence, Wenzhou University of Technology, Wenzhou, China}

\author{Yang Xiao}
\email{Corresponding author, xiaoyangphy@gmail.com}
\affiliation{
School of Physics, Henan Normal University,
Xinxiang 453007, P. R. China
}

\date{\today}

\begin{abstract}
False vacuum decay is conventionally described by a spherically symmetric bounce with a single negative fluctuation mode. This raises a fundamental question: can the single negative mode requirement itself enforce spherical symmetry, independently of action minimization?
For a single real scalar field with a smooth potential in flat spacetime, we rigorously prove that every localized nonspherical saddle with zero initial velocity has at least three negative modes. Additional reflection symmetries further tighten these bounds, as illustrated in a sextic scalar model. Our findings exclude nonspherical bounce candidates with a single negative mode, providing a complementary spectral foundation for the spherical bounce description of false vacuum decay.
\end{abstract}

\maketitle

\section{Introduction}
False vacuum decay through the nucleation of true vacuum bubbles is a fundamental nonperturbative process in quantum field theory~\cite{Devoto:2022qen,Coleman:1977py,Callan:1977pt}. It has important implications ranging from particle cosmology to condensed matter physics~\cite{Caprini:2019egz,Yin:2024hjm,Lagnese:2021grb,Athron:2023xlk,Mazumdar:2018dfl}. In the early Universe, first-order phase transitions can generate a stochastic gravitational wave background that may be detectable by future space-based observatories such as LISA~\cite{amaro2017laser}, TianQin~\cite{luo2016tianqin} and Taiji~\cite{ruan2020taiji}. Meanwhile, recent experiments with ultracold atomic gases have demonstrated analogues of false vacuum decay through bubble nucleation, offering new opportunities to investigate this process in a laboratory environment~\cite{Zenesini:2023afv,Cominotti:2025qia}. These developments underscore the importance of a deeper theoretical understanding of vacuum bubble nucleation.

The semiclassical description of true vacuum bubble nucleation is based on bounce solutions of the Euclidean field equations~\cite{Coleman:1977py,Callan:1977pt,rubakov2009classical}. Since the nucleation rate depends exponentially on the Euclidean action, the solution with the least action plays a central role in determining the dominant tunneling process. For a broad class of scalar potentials, the least-action solution of a single scalar field was shown to be spherically symmetric~\cite{Coleman:1977th}. This result was subsequently extended to the case of multiple scalar fields~\cite{Blum:2016ipp}. The framework was also generalized to finite temperatures, where nucleation is described by three dimensional critical bubbles in the high temperature limit~\cite{Linde:1981zj,Linde:1980tt}. These developments have established the spherical bounce picture as the basis of nucleation calculations. Notably, a critical bubble can either expand or collapse under perturbations, illustrating its instability against changes in size. This provides a physical interpretation of the negative fluctuation mode of the bounce. Coleman further emphasized that a conventional tunneling bounce is characterized by one and
only one negative fluctuation mode~\cite{Coleman:1987rm}.

Strictly speaking, the dominance of the spherical bounce does not preclude the existence of nonspherical bubbles. In particular, nonspherical bubble nucleation can occur in the presence of black holes or topological defects~\cite{Hiscock:1995ma,Gregory:2013hja,Gorbunov:2026wrz,Dasgupta:1997kn,Chatrchyan:2025uar,Firouzjahi:2020hfq}. Quantum and thermal fluctuations can also induce aspherical deformations of expanding bubbles, allowing even isolated bubbles to emit gravitational waves and potentially modifying the high-frequency part of the gravitational wave spectrum from first-order phase transitions~\cite{Bian:2026xdm,Blum:2024hcs}. Interestingly, even in homogeneous flat spacetime without any external source of symmetry breaking, the existence of nonspherical Euclidean solutions had already been established through variational methods~\cite{bartsch1993infinitely,mederski2020nonradial}. Guided by the symmetry structure identified in these studies, a regular \(O(4)\)-breaking solution was recently constructed in a sextic scalar model~\cite{Avraham:2026mgk}. Although the reported solution possesses multiple negative modes at a representative benchmark, it remains an open question whether nonspherical bounce solutions with exactly one negative mode can exist in homogeneous flat spacetime and, if so, what their collective contribution to the vacuum decay rate would be.

In this work, we take a first step toward answering this question by rigorously proving that every localized nonspherical Euclidean saddle satisfying the zero velocity matching condition has at least three negative modes. This result rules out the existence of nonspherical bubbles with exactly one negative mode for a single real scalar field in homogeneous flat spacetime. The remainder of this paper is organized as follows. In Sec.~\ref{sec:II}, we establish the negative-mode theorem for nonspherical Euclidean saddles. In Sec.~\ref{sec:III}, we apply the theorem to a sextic scalar model and discuss its relation to previous studies. Finally, we summarize our conclusions in Sec.~\ref{sec:IV}.

\section{Negative-Mode Bounds for Nonspherical Saddles}\label{sec:II}

The Euclidean action of a single real scalar field in
four-dimensional flat spacetime is given by
\begin{equation}
    S_E[\phi]=\int_{\mathbb R^4}d^4x
    \left[
        \frac12(\partial_\mu\phi)^2+V(\phi)
    \right],
    \label{eq:euclidean_action}
\end{equation}
where $x=(\tau,\mathbf y)$, with $\tau$ denoting Euclidean
time and $\mathbf y\in\mathbb R^3$ the spatial coordinates.
The corresponding Euclidean field equation reads
\begin{equation}
    -\Delta \phi+V'(\phi)=0.
    \label{eq:euclidean_eom}
\end{equation}
Here $\Delta=\partial_\tau^2+\nabla_{\mathbf y}^2$
is the four-dimensional Euclidean Laplacian. We assume that $V\in C^3(\mathbb R)$ possesses at least
two local minima, corresponding to the false and true
vacua. Without loss of generality, we place the false
vacuum at $\phi=0$ and normalize the potential such
that $V(0)=V'(0)=0$. We further assume $V''(0)=m^2>0$.

Bounce solutions are nontrivial, bounded, localized
solutions $\phi_b\in C^2(\mathbb R^4)\cap H^1(\mathbb R^4)$
of Eq.~\eqref{eq:euclidean_eom}, satisfying the
boundary conditions
\begin{align}
    \phi_b(x)&\to0
    \qquad (|x|\to\infty), \notag\\
    \partial_\tau\phi_b(0,\mathbf y)&=0
    \qquad (\forall\,\mathbf y\in\mathbb R^3).
    \label{eq:bc_turning}
\end{align}
The second condition specifies a turning point in
Euclidean time, chosen at $\tau=0$, where the field
can be matched to a real-time configuration with
vanishing initial velocity~\cite{rubakov2009classical}.

To examine the stability properties of these solutions,
we consider small fluctuations $\phi_b\to\phi_b+\eta$.
Expanding the Euclidean action around a stationary
solution gives
\begin{equation}
    S_E[\phi_b+\eta]
    =S_E[\phi_b]+\frac12 Q_{\phi_b}[\eta]
    +O(\eta^3),
\end{equation}
where the quadratic form is
\begin{equation}
    Q_{\phi_b}[\eta]
    =\int_{\mathbb R^4}d^4x
    \left[
        (\partial_\mu\eta)^2+V''(\phi_b)\eta^2
    \right].
    \label{eq:quadratic_form}
\end{equation}
The corresponding fluctuation operator is
\begin{equation}
    \mathcal M_{\phi_b}
    =-\Delta +V''(\phi_b).
    \label{eq:fluctuation_operator}
\end{equation}
In the semiclassical description,
the tunneling bounce is characterized by
exactly one negative fluctuation
mode~\cite{Coleman:1987rm}. The number of such modes,
corresponding to negative eigenvalues of
$\mathcal M_{\phi_b}$, defines the Morse index of the
solution, denoted by $\operatorname{ind}(\phi_b)$.
Thus, a conventional tunneling bounce must satisfy
$\operatorname{ind}(\phi_b)=1$.

An equivalent mathematical characterization of the
Morse index is the maximal dimension of a linear
subspace $E\subset H^1(\mathbb R^4)$ for which
$Q_{\phi_b}[\eta]<0$ for every nonzero $\eta\in E$.
The Morse index can therefore be expressed as
\begin{equation}
    \operatorname{ind}(\phi_b)
    =\sup_{\substack{
        E\subset H^1(\mathbb R^4)\\
        Q_{\phi_b}|_E<0
    }}\dim E.
    \label{eq:morse_index}
\end{equation}
So the central question is whether the maximal
dimension of a negative-definite subspace can be
one for a nonspherical candidate bounce.

The key to answering this question lies in the translational zero mode $w=\partial_\tau\phi_b$, which satisfies $\mathcal M_{\phi_b}w=0$. As shown below, the turning condition implies
reflection symmetry in Euclidean time, while nonsphericity forces $w$ to change sign in the positive Euclidean-time half-space. Its positive and negative parts can then be used to construct independent directions with vanishing
quadratic form in both time-parity sectors. Although these directions have vanishing quadratic form, not all of them can be zero modes. The remaining independent directions impose a lower bound on the dimension of the negative eigenspace. We now make this argument precise.

\begin{lemma}[Sign-changing property of nonspherical solutions]
\label{lem:one_sign}
Under the above assumptions, every nontrivial stationary
solution that does not change sign is radially symmetric
about some point in $\mathbb R^4$. Consequently, every
nonspherical stationary solution must change sign.
\end{lemma}

\begin{proof}
Suppose $\phi_b\geq0$. Since $V'(0)=0$, the field
equation can be written as
$-\Delta\phi_b+c(x)\phi_b=0$, where
$c(x)=V'(\phi_b)/\phi_b$ extends continuously
to $\phi_b=0$. By the strong maximum principle,
a nontrivial nonnegative solution cannot vanish
at any finite point, implying $\phi_b>0$
throughout $\mathbb R^4$~\cite{gidas1979symmetry}.

Writing the field equation as
$\Delta\phi_b+g(\phi_b)=0$, with $g=-V'$,
we have $g'(0)=-m^2<0$. Thus, $g$ is decreasing
near zero. The moving-plane theorem then implies that the positive
solution $\phi_b$, which vanishes at infinity,
is radially symmetric about some point in
$\mathbb R^4$~\cite{Li01011993} .

For $\phi_b\leq0$, the same argument applies
to $v=-\phi_b\geq0$, which satisfies
$\Delta v+g_-(v)=0$, with $g_-(v)=V'(-v)$
and $g_-'(0)=-m^2<0$. Hence, every nontrivial
stationary solution that does not change sign
must be radially symmetric. By contraposition,
every nonspherical stationary solution must
change sign.
\end{proof}

\begin{lemma}[Time-reflection symmetry]
\label{lem:reflection}
Any solution $\phi_b$ satisfying the Euclidean field
equation and the turning condition
$\partial_\tau\phi_b(0,\mathbf y)=0$ is symmetric
under Euclidean-time reflection,
\begin{equation}
    \phi_b(-\tau,\mathbf y)
    =\phi_b(\tau,\mathbf y).
    \label{eq:reflection_symmetry}
\end{equation}
\end{lemma}
\begin{proof}
Consider the reflected field
$\widetilde{\phi}_b(\tau,\mathbf y)
=\phi_b(-\tau,\mathbf y)$, which satisfies the
same Euclidean field equation as $\phi_b$.
Their difference,
$h=\phi_b-\widetilde{\phi}_b$, obeys
\begin{equation}
    -\Delta h+c(x)h=0,
    \qquad
    c(x)=\int_0^1
    V''\bigl(\widetilde{\phi}_b+s h\bigr)\,ds.
    \label{eq:reflection_difference}
\end{equation}
Since $V\in C^3(\mathbb R)$ and $\phi_b$ is
bounded, the coefficient $c(x)$ is bounded.
By construction and the turning condition,
\begin{equation}
    h(0,\mathbf y)=0,
    \qquad
    \partial_\tau h(0,\mathbf y)=0.
\end{equation}
The unique continuation for elliptic equations implies that $h$ vanishes identically~\cite{klein2016quantitative}. Consequently,
\begin{equation}
    \phi_b(-\tau,\mathbf y)
    =\phi_b(\tau,\mathbf y),
\end{equation}
establishing reflection symmetry in Euclidean time.
\end{proof}

The two lemmas together constrain the sign
structure of the translational zero mode.
Differentiating Eq.~\eqref{eq:euclidean_eom}
with respect to $\tau$ gives
\begin{equation}
    w=\partial_\tau\phi_b,
    \qquad \mathcal M_{\phi_b}w=0.
    \label{eq:translation_zero_mode}
\end{equation}
By Lemma~\ref{lem:reflection}, $\phi_b$ is even
in $\tau$, and hence $w$ is odd.
For a nonspherical solution,
Lemma~\ref{lem:one_sign} ensures that $\phi_b$
changes sign. Together with time-reflection
symmetry, this implies that both positive and
negative field values occur in the half-space
$\tau>0$. Since $\phi_b(\tau,\mathbf y)\to0$
as $\tau\to\infty$, positive field values must
decrease somewhere along the Euclidean-time
direction, while negative values must increase.
Consequently, $w$ takes both signs in $\tau>0$, as illustrated schematically in Fig.~\ref{fig:sign_change}.
\begin{figure*}[t]
    \centering
    \includegraphics[width=\textwidth]{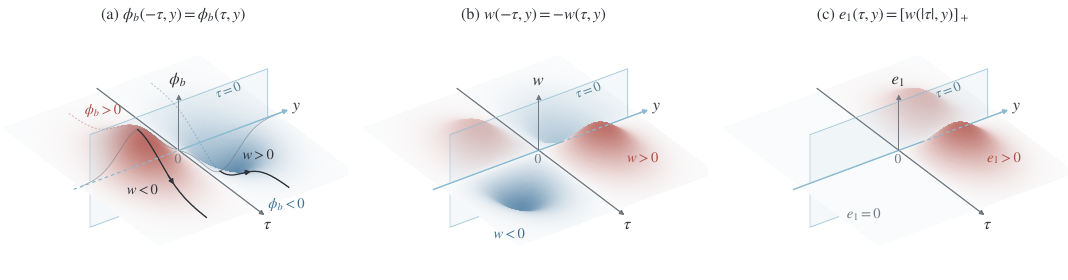}
    \caption{Schematic illustration of the nodal construction underlying our main theorem. (a) \textit{Left}: Sign-changing and time-reflection nature of a nonspherical saddle (see Lemmata), which guarantees the translational zero mode $w=\partial_\tau\phi_b$ takes both signs in the half-space $\tau>0$. (b) \textit{Middle}: The $\tau$-odd translational zero mode, whose positive and negative nodal domains allow the construction of independent trial fluctuations. (c) \textit{Right}: The $\tau$-even trial fluctuation $e_1=[w(|\tau|,y)]_+$ obtained by reflecting the positive part of $w$ across the turning slice. Although $Q[e_1]=0$, it is not a zero mode of the Hessian. However, if $m_+\leq 1$, a nontrivial combination of $e_1$ and its negative-domain counterpart $e_2$ would necessarily be a zero mode, contradicting unique continuation.}
    \label{fig:sign_change}
\end{figure*}

To exploit this sign structure, we decompose
the fluctuation space according to Euclidean-time
parity. Any fluctuation $\eta\in H^1(\mathbb R^4)$
can be uniquely written as $\eta=\eta_++\eta_-$,
where
\begin{equation}
    \eta_\pm(\tau,\mathbf y)
    =\frac{1}{2}
    \left[
        \eta(\tau,\mathbf y)
        \pm\eta(-\tau,\mathbf y)
    \right].
    \label{eq:parity_decomposition}
\end{equation}
The even and odd components belong to the
subspaces
\begin{equation}
    \mathcal H_\pm=
    \left\{
        \eta\in H^1(\mathbb R^4):
        \eta(-\tau,\mathbf y)
        =\pm\eta(\tau,\mathbf y)
    \right\},
    \label{eq:parity_spaces}
\end{equation}
which form the direct sum
$H^1(\mathbb R^4)=\mathcal H_+\oplus\mathcal H_-$.

Since $\phi_b$ is even in $\tau$, so is
$V''(\phi_b)$. The fluctuation operator
$\mathcal M_{\phi_b}$ therefore preserves
each parity sector. Moreover, all cross terms
between $\eta_+$ and $\eta_-$ in the quadratic
form are odd in $\tau$ and vanish upon integration
over $\mathbb R^4$. Thus,
\begin{equation}
    Q_{\phi_b}[\eta]
    =Q_{\phi_b}[\eta_+]
    +Q_{\phi_b}[\eta_-].
    \label{eq:quadratic_decomposition}
\end{equation}
The full Morse index consequently decomposes
as $\operatorname{ind}(\phi_b)=m_++m_-$,
where $m_\pm$ denote the maximal dimensions
of negative-definite subspaces of $Q_{\phi_b}$
within $\mathcal H_\pm$, respectively. With these ingredients in place, we can now state our main theorem.

\begin{theorem}[Negative-mode bounds, \textbf{main}]
\label{thm:main}
Under the assumptions stated above, every nonspherical
candidate bounce $\phi_b$ satisfies
\begin{equation}
    m_+\geq2,\qquad m_-\geq1,
    \qquad \operatorname{ind}(\phi_b)\geq3.
    \label{eq:main_bound}
\end{equation}
\end{theorem}

\begin{proof}
Let $D=\{(\tau,\mathbf y):\tau>0\}$ and separate
the translational zero mode into its positive
and negative parts on $D$,
\begin{equation}
    p_1=\max\{w,0\},\qquad
    p_2=\max\{-w,0\}.
    \label{eq:zero_mode_parts}
\end{equation}
Both are nonzero because $w$ changes sign in $D$.
Multiplying $\mathcal M_{\phi_b}w=0$ by $p_1$
and $-p_2$, respectively, and integrating over $D$
gives $Q_D[p_1]=Q_D[p_2]=0$, where $Q_D$ denotes
the quadratic form restricted to $D$. 
A technical subtlety arises from the possible
nonsmoothness of $p_1$ and $p_2$ across the nodal
set of $w$. Their admissibility as weak test
functions is justified in Sec.~\ref{sec:supp_truncation} of the \textit{Supplement Material}.

We extend $p_1$ and $p_2$ from the half-space
$D$ to $\mathbb R^4$ by even and odd reflection
across $\tau=0$:
\begin{align}
    e_i(\tau,\mathbf y)
    &=p_i(|\tau|,\mathbf y),\notag\\
    o_i(\tau,\mathbf y)
    &=\operatorname{sgn}(\tau)p_i(|\tau|,\mathbf y),
    \qquad i=1,2.
    \label{eq:parity_extensions}
\end{align}
Consequently, $e_i(-\tau,\mathbf y)=e_i(\tau,\mathbf y)$
and $o_i(-\tau,\mathbf y)=-o_i(\tau,\mathbf y)$.
Since $p_i$ vanishes on the turning slice,
both extensions belong to $H^1(\mathbb R^4)$.
Thus, $e_i\in\mathcal H_+$ and $o_i\in\mathcal H_-$.
For illustration, a schematic profile of the trail function $e_1$ is shown in Fig.~\ref{fig:sign_change}.

These extensions provide two linearly independent
fluctuations in each parity sector: $e_1,e_2$
in $\mathcal H_+$ and $o_1,o_2$ in $\mathcal H_-$.
Their independence follows from the fact that
$p_1$ and $p_2$ are nonzero in disjoint sign
regions of $w$. Moreover, the corresponding
mixed terms in the quadratic form vanish,
since the two parts and their gradients do not
overlap almost everywhere. Combining this with
$Q_D[p_1]=Q_D[p_2]=0$ and reflection symmetry,
we obtain
\begin{equation}
    Q_{\phi_b}[ae_1+be_2]
    =Q_{\phi_b}[ao_1+bo_2]=0
    \label{eq:zero_form_combinations}
\end{equation}
for arbitrary real $a,b$. Thus, each parity
sector contains a two-dimensional subspace
on which the quadratic form vanishes identically.
Importantly, these fluctuations need not be
zero modes of $\mathcal M_{\phi_b}$.

To relate these directions to the negative-mode
spectrum, we invoke a standard consequence
of the spectral theorem: in either parity sector,
the quadratic form is nonnegative on the subspace
orthogonal to all negative modes, and it vanishes
there only for zero modes. The spectral
justification is given in Sec.~\ref{sec:supp_spectral} of the \textit{Supplement Material}.

Consider first the even sector and suppose,
for contradiction, that $m_+\leq1$. This means
that there is at most one linearly independent
negative mode in the even sector. Since $e_1$
and $e_2$ span a two-dimensional space, we can
always find a nonzero combination $z=ae_1+be_2$
that is orthogonal to this negative mode, if
present. By Eq.~\eqref{eq:zero_form_combinations},
$Q_{\phi_b}[z]=0$. The spectral property
stated above then implies that
$\mathcal M_{\phi_b}z=0$.

However, such a zero mode cannot exist.
In the nonempty region where $w>0$ and
$\tau>0$, we have $e_1=w$ and $e_2=0$,
so $z=aw$. Since both $z$ and $w$ satisfy
the zero-mode equation, unique continuation for elliptic equations extends this equality to all of $\mathbb R^4$.
This is impossible because $z$ is even
whereas $w$ is odd under time reflection,
forcing $z\equiv0$, in contradiction with
our choice of a nonzero $z$. We therefore
conclude that $m_+\geq2$.

We next consider the odd sector and suppose,
for contradiction, that $m_-=0$, so that no
odd negative modes exist. Since $o_1$ is a
nonzero odd fluctuation satisfying
$Q_{\phi_b}[o_1]=0$, the spectral property
stated above implies that
$\mathcal M_{\phi_b}o_1=0$.

This conclusion is incompatible with the
construction of $o_1$. By definition, $o_1$
vanishes in the nonempty open region where
$w<0$ and $\tau>0$, while remaining nonzero
where $w>0$. Since $o_1$ satisfies the
zero-mode equation, unique continuation
would force it to vanish identically
throughout $\mathbb R^4$, yielding a
contradiction. Hence $m_-\geq1$.

Combining the two sectors, we obtain
$\operatorname{ind}(\phi_b)=m_++m_-\geq3$,
thereby ruling out nonspherical bounce
candidates with only one negative mode.
\end{proof}

For comparison, the conventional
$\mathrm O(4)$-symmetric bounce possesses
a single negative mode associated with a
radial fluctuation. This mode is even under
Euclidean time reflection and therefore belongs
to $\mathcal H_+$, giving $(m_+,m_-)=(1,0)$.
Additional spatial reflection symmetries of
nonspherical solutions can further strengthen
the bounds established above, as stated in
the following corollary.

\begin{corollary}[Additional reflection symmetries]
\label{cor:reflection}
Suppose that a nonspherical solution $\phi_b$
satisfying the assumptions of
Theorem~\ref{thm:main} is invariant under
reflections across $k$ mutually orthogonal
hyperplanes, including the turning slice
$\tau=0$. Then
\begin{equation}
    m_+\geq k+1,\qquad
    m_-\geq1,\qquad 1\leq k\leq4,
    \label{eq:reflection_bound}
\end{equation}
and hence $\operatorname{ind}(\phi_b)\geq k+2$.
\end{corollary}

\begin{proof}
Since $w=\partial_\tau\phi_b$ is even under
all spatial reflections, the fluctuations
$e_1$ and $e_2$ constructed in
Theorem~\ref{thm:main} belong to the sector
even under all $k$ reflections. The same
argument therefore guarantees at least
two negative modes in this sector.
For each of the $k-1$ additional spatial
mirrors, the corresponding translational
zero mode is odd under that reflection
and even under all others. The sign-truncation
and spectral arguments then establish at
least one negative mode in each distinct
time-even parity sector. Since these sectors
are mutually orthogonal, we obtain
$m_+\geq2+(k-1)=k+1$.
The bound $m_-\geq1$ follows from
Theorem~\ref{thm:main}, with further details
provided in Sec.~\ref{sec:supp_reflections} of the \textit{Supplement Material}.
\end{proof}

\section{Application: A Sextic Scalar Model} \label{sec:III}
As a direct application of our main theorem,
we consider the sextic scalar model studied
in Ref.~\cite{Avraham:2026mgk}, with the potential
\begin{equation}
 V_\gamma(\phi)=\frac12\phi^2-\frac14\phi^4+
 \frac{\gamma}{6}\phi^6.\label{eq:sextic}
\end{equation}
For $0<\gamma<3/16$, the origin $\phi=0$ is a metastable vacuum with positive mass squared, while two degenerate nonzero minima lie at lower potential energy. 

This model admits nonspherical, sign-changing
Euclidean saddles with $O(2)\times O(2)$ symmetry,
as follows from the variational theorem of
Ref.~\cite{mederski2020nonradial} together with the
high-field truncation procedure described therein. A solution with these symmetries was recently constructed numerically~\cite{Avraham:2026mgk}. Writing $x=(x_1,x_2,x_3,x_4)\in\mathbb R^4$, the symmetry implies that its field profile takes the form 
\begin{equation}
    \phi_b(x) = \phi_b(r,s),\quad r=\sqrt{x_1^2+x_2^2},\quad s=\sqrt{x_3^2+x_4^2},
\end{equation}
with the additional exchange antisymmetry $\phi_b(r,s)=-\phi_b(s,r)$. Consequently, the solution is invariant under four mutually orthogonal Cartesian reflections, each providing an exact flat turning slice upon choosing the corresponding coordinate as Euclidean time. Applying Corollary~\ref{cor:reflection} with $k=4$ gives
\begin{equation}\label{eq:sextic_bound}
 m_+\ge5,\qquad m_-\ge1,
 \qquad \ind_{\mathrm{full}}\ge6. 
\end{equation}

At $\gamma=0.1$, Ref.~\cite{Avraham:2026mgk} reported 15 negative modes, comprising three $(0,0)$ modes and three fourfold-degenerate multiplets $(0,n)\oplus(n,0)$ with $n=1,2,3$. Here, labels $(\ell_1,\ell_2)$ denote the angular Fourier indices in two orthogonal planes. Under $x_1$-reflection, the $(0,0)$ modes are all even, while each fourfold-degenerate multiplet contains three even modes and one odd mode. The reported spectrum therefore yields 
\begin{equation}
\begin{aligned}
    m_{+,\text{reported}} &= 12 > 5, \\
    m_{-,\text{reported}} &= 3 > 1, \\
    \ind_{\mathrm{full, reported}} &= 15 > 6.
\end{aligned}
\end{equation}
This parity decomposition is consistent with our analytic bounds, which hold throughout the whole range $0<\gamma<3/16$ without knowledge of the explicit saddle or its fluctuation spectrum. In particular, the symmetry and nodal structure enforce multiple time-even negative modes, ruling out the conventional single-negative-mode spectrum even when fluctuations are restricted to preserve the turning slice.

\section{Conclusion}\label{sec:IV}

In this work, we have shown that localized nonspherical Euclidean saddles of a single real scalar field in flat space have at least three negative modes under the regularity and zero-velocity matching conditions specified above. This establishes a \textit{no-go} statement for interpreting nonspherical saddles as Coleman-type bounces governing false-vacuum decay. 

This distinction is structural rather than merely quantitative. In Coleman's semiclassical description, saddles with additional negative modes fall outside the conventional tunneling interpretation \cite{Coleman:1987rm}. 
This distinction is further illustrated by studies of oscillating gravitational instantons, where additional negative modes have motivated an alternative interpretation in terms of unstable intermediate thermal configurations \cite{Lavrelashvili:2006cv, Battarra:2012vu}. From a broader Morse-theoretic perspective, higher-index saddles also reveal a richer critical-point structure of the Euclidean action landscape \cite{Witten:1982im, yin2020construction}. While these studies suggest possible roles for higher-index saddles beyond conventional tunneling, their physical interpretation requires an analysis
beyond the local Euclidean fluctuation spectrum,
for example through real-time functional
methods~\cite{Ai:2019fri}.

Earlier studies have established important connections between symmetry, nodal structure, and Morse index in bounded radial domains \cite{aftalion2004qualitative} and for biradial solutions \cite{abatangelo2011solutions, Avraham:2026mgk,abatangelo2010note}. Going beyond these symmetry-restricted or numerical results, this work establishes parity-resolved lower bounds on the full Morse index at a prescribed turning slice. The parity resolution is essential: time-even fluctuations preserve the zero-velocity matching condition, whereas a bound on the total Morse index alone would not exclude a single negative direction in this sector. By ruling out this possibility, our results provide a model-independent spectral criterion for nonspherical tunneling candidates. 

A broader question is what physical role, if any, higher-index nonspherical Euclidean saddles play in quantum tunneling.  Clarifying the possible significance of these saddles may offer new insights into quantum tunneling beyond the conventional single-bounce picture.

\begin{acknowledgments}
This work was supported by the National Natural Science Foundation of China (Grant No.~12605177),  the National Natural Science Foundation of Henan province (Grant No.~27AQ140008).
\end{acknowledgments}

\bibliographystyle{apsrev4-2}
\bibliography{references}

\clearpage
\onecolumngrid
\section*{Supplemental Material}

\setcounter{section}{0}
\renewcommand{\thesection}{S\arabic{section}}
\setcounter{equation}{0}
\renewcommand{\theequation}{S\arabic{equation}}

\section{Weak truncation of the translational zero mode}
\label{sec:supp_truncation}

The zero-mode construction in Theorem~\ref{thm:main}
uses the positive and negative parts of $w$ on
$D=\{\tau>0\}$. Their possible lack of smoothness at
the zeros of $w$ does not obstruct the argument: the
zero-mode equation can be tested in the weak sense, with
no assumption that the nodal boundaries are smooth.

Because $w$ is odd under time reflection, it has zero
trace on $\partial D=\{\tau=0\}$. The regularity
$w\in H^1$, used below, follows from the elliptic
estimates in Sec.~\ref{sec:supp_spectral}.
The positive and negative parts
$p_1=\max\{w,0\}$ and $p_2=\max\{-w,0\}$
belong to $H_0^1(D)$. The Sobolev chain rule gives
\begin{equation}
    \nabla p_1=\mathbf 1_{\{w>0\}}\nabla w,
    \qquad
    \nabla p_2=-\mathbf 1_{\{w<0\}}\nabla w
    \quad\text{a.e. in }D.
    \label{eq:supp_truncation_gradients}
\end{equation}
Here $\mathbf 1_A$ denotes the indicator function of
$A$. These identities show that $p_1$ and $p_2$, as
well as their weak gradients, have no overlap almost
everywhere. Testing $\mathcal M_{\phi_b}w=0$ with
$p_1$ and $-p_2$, respectively, therefore gives
\begin{equation}
    Q_D[p_1]=Q_D[p_2]=0,
    \label{eq:supp_halfspace_forms}
\end{equation}
with vanishing mixed quadratic term.

The zero trace also ensures that the even and odd
extensions $e_i,o_i$ of
Eq.~\eqref{eq:parity_extensions} belong to
$H^1(\mathbb R^4)$. Reflection symmetry doubles the
half-space quadratic form, so for any real $a,b$,
\begin{equation}
    Q_{\phi_b}[ae_1+be_2]
    =Q_{\phi_b}[ao_1+bo_2]=0.
    \label{eq:supp_extended_forms}
\end{equation}
It is important that $p_1,p_2$ are the \emph{full}
positive and negative parts on $D$, rather than
restrictions to selected connected nodal domains.
This choice preserves any additional spatial reflection
symmetries of $w$ and is therefore also needed for the
strengthened bound in Sec.~\ref{sec:supp_reflections}.

\section{Time reflection and spectral properties}
\label{sec:supp_spectral}

We collect two facts used in Sec.~II: the turning condition
implies time-reflection symmetry, and the quadratic-form
argument can be interpreted in terms of genuine negative
and zero modes. These facts also explain why the
sign-truncated fluctuations used in the proof are sufficient.

Let $\widetilde\phi_b(\tau,\mathbf y)
=\phi_b(-\tau,\mathbf y)$ and
$h=\phi_b-\widetilde\phi_b$. Both fields obey the same
Euclidean equation, so their difference satisfies
\begin{equation}
    (-\Delta+c(x))h=0,\qquad
    c(x)=\int_0^1 V''(\widetilde\phi_b+s h)\,ds.
    \label{eq:supp_reflection_eq}
\end{equation}
The coefficient $c$ is bounded because $\phi_b$ is
bounded and $V$ is smooth. By construction,
$h(0,\mathbf y)=0$; the turning condition also gives
$\partial_\tau h(0,\mathbf y)=0$. Extend $h$ from the
positive-time side of a small neighborhood of this plane
by zero to the negative-time side. Since both its value and
normal derivative vanish on the plane, this extension
introduces no interface term and still solves a linear
elliptic equation weakly. It vanishes on an open set, so
unique continuation for elliptic operators with bounded
potentials~\cite{klein2016quantitative} makes it vanish throughout the neighborhood. Applying the same property to $h$ on the connected space $\mathbb R^4$ yields $h\equiv0$.
This establishes the reflection property without assuming
that $V$ is analytic.

We next justify the spectral step in the main proof.
Since $V'(0)=0$, the bounded range of $\phi_b$ gives
$|V'(\phi_b)|\leq C|\phi_b|$. The field equation and
$\phi_b\in H^1$ then imply $\Delta\phi_b\in L^2$.
Standard elliptic regularity yields $\phi_b\in H^2$ and,
in particular, $w=\partial_\tau\phi_b\in H^1$.
Furthermore, writing
\begin{equation}
    \mathcal M_{\phi_b}
    =(-\Delta+m^2)+W(x),\qquad
    W(x)=V''(\phi_b(x))-m^2,
    \label{eq:supp_operator}
\end{equation}
we have a bounded $W$ that tends to zero at infinity.
Thus $\mathcal M_{\phi_b}$ is self-adjoint on $H^2$,
its essential spectrum begins at $m^2>0$, and its
negative eigenspace is finite-dimensional.

The consequence needed in the proof is simple.
In any reflection-parity sector, a fluctuation orthogonal
in $L^2$ to all negative modes has no negative spectral
component, so its quadratic form is nonnegative. If this
quadratic form also vanishes, the fluctuation must belong
to $\ker\mathcal M_{\phi_b}$. This is precisely the
spectral-theorem step used to turn a hypothetical shortage
of negative modes into a zero mode, which can then be
excluded by unique continuation.

\section{Additional reflection symmetries}
\label{sec:supp_reflections}

We now justify the stronger bound in
Corollary~\ref{cor:reflection}. Suppose the solution is
invariant under $k$ mutually orthogonal reflections,
including time reflection $R_0$, and denote the other
reflections by $R_j$ ($j=1,\ldots,k-1$). Since these
reflections commute with $\mathcal M_{\phi_b}$,
fluctuations can be grouped into joint-parity sectors.
In particular, let $\mathcal K_{\rm all}$ be the sector
even under every reflection, and $\mathcal K_j$ the
sector odd under $R_j$ but even under all the others.
All these sectors are time-even and are mutually
orthogonal for the quadratic form.

The first contribution comes from the jointly even
sector. The time-translation zero mode
$w=\partial_\tau\phi_b$ is even under every spatial
reflection. Its full positive and negative parts, and
hence their time-even extensions $e_1,e_2$, share
these spatial symmetries. Consequently,
$e_1,e_2\in\mathcal K_{\rm all}$, and their span has
identically vanishing quadratic form by
Eq.~\eqref{eq:supp_extended_forms}. If
$\mathcal K_{\rm all}$ had at most one negative mode,
a nonzero combination $z=ae_1+be_2$ could be
chosen orthogonal to all of them. The spectral argument
of Sec.~\ref{sec:supp_spectral} would then make $z$
a genuine zero mode. On the open set
$\{\tau>0,\,w>0\}$, however, $z=aw$.
Unique continuation would extend this equality to
all of $\mathbb R^4$, contrary to the fact that $z$
is time-even while $w$ is time-odd. Thus
$\mathcal K_{\rm all}$ contains at least two
negative modes.

Each additional spatial reflection supplies another
independent sector. Let $\mathbf n_j$ be the unit
normal to its reflection plane and choose a point $c$
in the common intersection of the planes. Translation
invariance of the field equation gives
\begin{equation}
    w_j=\mathbf n_j\cdot\nabla\phi_b,
    \qquad \mathcal M_{\phi_b}w_j=0.
    \label{eq:supp_spatial_zero_mode}
\end{equation}
This zero mode is odd under $R_j$ and even under
the remaining reflections, including $R_0$.
As in the time direction, the sign change of $\phi_b$,
its reflection symmetry, and its approach to zero at
spatial infinity imply that $w_j$ takes both signs on
one side of the corresponding plane.

On that half-space, take the positive part of $w_j$
and extend it oddly across the plane, calling the result
$q_j$. The same weak-truncation argument as in
Sec.~\ref{sec:supp_truncation} shows that $q_j$ is
nonzero, belongs to $\mathcal K_j$, and satisfies
$Q_{\phi_b}[q_j]=0$. Were there no negative mode
in $\mathcal K_j$, the spectral theorem would make
$q_j$ a zero mode. But $q_j$ vanishes on the nonempty
open region where $w_j<0$ on the chosen half-space,
contradicting unique continuation. Hence each of the
$k-1$ sectors $\mathcal K_j$ contains at least one
negative mode.

The jointly even sector therefore contributes at least
two negative directions, while the $k-1$ distinct
spatial-odd sectors contribute one each. Their mutual
orthogonality allows these contributions to be added,
giving
\begin{equation}
    m_+\geq 2+(k-1)=k+1.
    \label{eq:supp_improved_bound}
\end{equation}
Together with $m_-\geq1$ from
Theorem~\ref{thm:main}, this proves
Corollary~\ref{cor:reflection}.

\end{document}